\documentclass[letterpaper, 10pt, conference,dvipsname]{ieeeconf}
\usepackage{xcolor}
\usepackage{cite}
\usepackage{graphicx}
\usepackage{header}
\usepackage{balance}
 \renewcommand{\QED}{\QEDopen}

\IEEEoverridecommandlockouts
\begin{document}
\title{\LARGE \bf Local Nash Equilibria are Isolated, Strict Local Nash Equilibria in `Almost All' Zero-Sum Continuous Games}

\author{Eric Mazumdar$^1$ and
\thanks{$^1$Department of Electrical Engineering and Computer Sciences, University of California, Berkeley, Berkeley, CA. email: {\tt mazumdar@berkeley.edu}}%
Lillian J. Ratliff$^2$\thanks{$^2$Department of Electrical and Computer Engineering, University of Washington, Seattle, WA. email: {\tt ratliffl@uw.edu}}
}
\maketitle

\begin{abstract}
 We prove that differential Nash equilibria are generic
amongst local Nash equilibria in continuous
zero-sum games. That is, there exists an open-dense subset of
zero-sum games for which local Nash equilibria are non-degenerate differential Nash
equilibria.
The result extends previous results to  the zero-sum setting, where we
obtain even stronger results; in particular, we show that local Nash equilibria
are generically hyperbolic critical points. 
We further show that differential Nash equilibria of zero-sum games are structurally stable. The purpose for presenting these
extensions is the recent renewed interest in zero-sum games within machine
learning and optimization. Adversarial learning and generative adversarial network approaches are
touted to be more robust than the alternative. Zero-sum games are at the heart
of such approaches. Many works proceed under the assumption of hyperbolicity of
critical points. Our results justify this assumption
by showing `almost all' zero-sum games admit local Nash equilibria that are
hyperbolic.
\end{abstract}

\IEEEpeerreviewmaketitle

\section{Introduction}
\label{sec:intro}

With machine learning algorithms increasingly being placed in more complex, real world settings, there has been a renewed interest in continuous games~\cite{ mertikopoulos:2019aa,zhang:2010aa, mazumdar:2018aa}, and particularly zero-sum continuous games~\cite{mazumdar:2019aa, daskalakis:2018aa, goodfellow:2014aa, jin:2019aa}. 
Adversarial learning~\cite{daskalakis:2017aa,mertikopoulos:2018aa}, robust reinforcement learning~\cite{li:2019aa, pinto:2017aa}, and generative adversarial networks~\cite{goodfellow:2014aa} all make use of zero-sum games played on highly non-convex functions to achieve remarkable results. 

Though progress is being made, a theoretical understanding of the equilibria of such games is lacking. 
In particular, many of the approaches to learning equilibria in these machine learning applications are gradient-based.
For instance, consider an adversarial learning setting where the goal is to learn a model or network by optimizing a function $f\in C^r(\Theta\times W,\mb{R})$ over $\theta\in \Theta$ where $w\in W$ is chosen by an adversary. A general approach to this problem is to study the coupled learning dynamics that arise when  one \emph{player} is descending $f$ and the other is ascending it---e.g.,
\begin{align*}
   \bmat{ \theta^+\\ w^+}=\bmat{\theta-\gamma D_\theta f(\theta,w)\\
w+\eta D_w f(\theta,w)}.
\end{align*}

%
Most convergence analysis depends on an assumption of \emph{local convexity in the game space} around an equilibrium---that is, nearby Nash equilibria the Jacobian of the gradient-based learning rule is assumed to be locally positive definite. 
Indeed, with respect to the above example, in consideration of the limiting dynamics
 $\dot{x}=-\omega(x)$ where $x=(\theta, w)$ and $\omega(x)=(D_\theta f(x), -D_wf(x))$,
many of the convergence guarantees in this setting proceed under the assumption that around critical points, the Jacobian 
\[J(\theta, w)=\bmat{D_{\theta}^2f(\theta,w) & D_{\theta w}f(\theta,w)\\ -D_{w \theta}f(\theta,w) & -D_w^2f(\theta, w)}\]
is positive definite---i.e., there is some notion of \emph{local convexity in the game space}.
Given the structural assumptions often invoked in the analysis of these learning approaches, two questions naturally arise: 
\begin{itemize}
    \item  Is this a `robust' assumption in the sense of \emph{structural stability}---i.e., does the property persist under smooth perturbations to the game?;
    \item Is this assumption satisfied for `almost all zero-sum games' in the sense of \emph{genericity}?
\end{itemize}
Building on the work in~\cite{ratliff:2013aa, ratliff:2014aa, ratliff:2016aa}, this paper addresses these two questions. 

Towards this end, we leverage a refinement of the local Nash equilibrium concept that defines an equilibrium in terms of the first- and second-order optimality conditions for each player holding all other players fixed.   This refinement has implicitly in its definition this notion of local convexity in the game space; it also has a structure that is particularly amenable to computation and which can be exploited in learning since it is characterized in terms of local information. Efforts to show this refinement is both structurally stable and generic aid in justifying its broad use.

\subsection{Contributions}
The contributions are summarized as follows: 

a. We prove that differential Nash equilibria---a refinement of local Nash equilibria defined in terms of first- and second-order conditions which characterize local optimality for players---are generic amongst local Nash equilibria in continuous zero-sum games (Theorem~\ref{thm:main}). This implies that almost all zero-sum games played on continuous functions admit local Nash equilibria that are strict and isolated.

b. Exploiting the underlying structure of zero-sum game---i.e., the game is defined in terms of a single sufficiently smooth cost function---we show that all differential Nash equilibria are hyperbolic (Proposition~\ref{prop:hyperbolic}), meaning they are locally exponentially attracting for gradient-play. Combining this fact with the above, we also show that local Nash equilibria are generically hyperbolic (Corollary~\ref{cor:hyperbolic}).

c. We prove that zero-sum games are structurally stable (Theorem~\ref{thm:ss_zsg}); that is, the structure of the game---and hence, its equilibria---is robust to smooth perturbations within the space of zero-sum games.

 In~\cite{ratliff:2013aa, ratliff:2014aa, ratliff:2016aa}, similar results to a.~and c.~were shown for the larger class of general-sum continuous games. Yet, the set of zero-sum games is of zero measure in the space of general-sum continuous games, and hence, the results of this paper are not a direct implication of those results.
 Further, b.~is a much stronger statement than the genericity result in~\cite{ratliff:2014aa}. In particular,~\cite{ratliff:2014aa} shows that non-degenerate differential Nash equilibria are generic amongst local Nash equilibria. We, on the other hand, show that in the class of zero-sum games, \emph{all differential Nash equilibria are non-degenerate, and moreover, hyperbolic}. The latter is a particularly strong result, achievable due to the specific structure of the zero-sum game. Indeed, two-player zero-sum continuous games are defined completely in terms of a single sufficiently smooth function---i.e., given $f\in C^r(X,\mb{R})$, the corresponding zero-sum game is $(f,-f)$.
%
%

Moreover, the work in this paper focuses on a class of games of particular import to the machine learning and robust control communities, where many recent works have made the assumption of hyperbolicity of local Nash equilibria without a thorough understanding of whether or not such an assumption is restrictive (see e.g. \cite{Daskalakis2018TheLP,balduzzi:2018aa,PotentialGameBanditFeed,VariationalGAN,jin:2019aa}). The results in this paper show that this assumption simply \emph{rules out a measure zero set of zero-sum games}. 


\section{Preliminaries}
 Before developing the main results, we present our general setup, as well as some preliminary game theoretic and mathematical definitions. Additional mathematical preliminaries are included in Appendix~\ref{app:prelims}.
\subsection{Preliminary Definitions}
\label{sec:setup}

In this paper, we consider full information continuous, two-player zero-sum games.
We use the term `player' and `agent' interchangeably. 
Each player $i\in \mc{I}=\{1,2\}$ selects an action $x_i$ from a
\emph{topological space} $X_i$ in order to minimize its cost $f_i:X\rar \mb{R}$
where $X=X_1\times X_2$ is the \emph{joint strategy space} of all the
agents. Note that $f_i$ depends on $x_{-i}$
which is the collection of actions of all other agents excluding agent
$i$---that is, $f_i:(x_i,x_{-i})\mapsto f_i(x_i,x_{-i})\in \mb{R}$.  Furthermore, each
$X_i$ can be finite-dimensional smooth manifolds or
infinite-dimensional Banach manifolds. Each player's cost function $f_i$ is
assumed to be sufficiently smooth. 

A two-player zero-sum game is characterized by a cost function $f\in C^r(X,\mb{R})$ in the sense that the first player aims to minimize $f$ with respect to $x_1$ and the second player aims to maximize $f$ with respect to $x_2$---that is, $f_1\equiv f$ and $\f_2\equiv -f$.
Hence, given a function $f$, we denote a two-player zero-sum game by $(f,-f)$ where $f\in C^r(X,\mb{R})$.

As is common in the study of games, we adopt the Nash equilibrium concept to
characterize the interaction between agents.
\begin{definition}
  \label{def:SLNE}
  A strategy $x=(\pxone{1}, \pxone{2})\in X$ is a {local Nash equilibrium}
  for the game $(f_1,f_2)=(\f, -\f)$ if for each $i\in\mc{I}$ there exists
  an open set $W_i\subset X_i$ such
  that $\pxone{i}\in W_i$ and 
  \[f_i(x_i,x_{-i})\leq f_i(x_i',x_{-i}), \ \ \forall\ \pxone{i'}\in
  W_i\bs\{\pxone{i}\}.\]
    If the above inequalities are strict, then we say
    $(\pxone{1},x_2)$ is a {strict local Nash equilibrium}.
    If $W_i=X_i$ for each $i$, then $(\pxone{1},x_2)$ is a {global Nash
  equilibrium}. 
\end{definition}

In~\cite{ratliff:2013aa} and subsequent works~\cite{ratliff:2014aa,ratliff:2016aa}, a refinement of the local Nash equilibrium concept known as a
\emph{differential Nash equilibrium} was introduced. This refinement
characterizes local Nash in terms of first- and second-order conditions on
player cost functions, and even in the more general non-convex setting, a
differential Nash equilibrium was shown to be well-defined and independent of
the choice of coordinates on $X$. Moreover, for general sum games, differential Nash were shown to be
generic and structurally stable in $n$-player games. 

Towards defining the differential Nash concept, we introduce the following
mathematical object. A \emph{differential game form} is a differential 1-form $\omega: X\rar T^\ast X$ defined by
\[\textstyle\omega=\psi_{X_1}\circ d\f-\psi_{X_2}\circ d\f\]
    where $\psi_{X_i}$ are the natural bundle maps 
    $\psi_{X_1}: T^\ast X\rar T^\ast X$
    that annihilate those components of the co-vector field $d\f$
    corresponding to $X_1$ and analogously for $\psi_{X_2}$.
    Note that when each $X_i$ is a finite dimensional manifold of dimensions $m_i$
(e.g., Euclidean space $\mb{R}^{m_i}$), then the differential game form has the 
coordinate representation,
\[ 
  \textstyle  \omega_\psi=\sum_{j=1}^{m_1}\frac{\partial( \f\circ
\psi^{-1})}{\partial y_1^j}dy_1^j+\sum_{j=1}^{m_2}\frac{\partial( -\f\circ
\psi^{-1})}{\partial y_2^j}dy_2^j,\]
for product chart $(U,\psi)$ in $X$ at $x=(x_1,\ldots, x_n)$ with local
coordinates $(y_1^1, \ldots, y_{1}^{m_1}, y_2^1, \ldots, y_2^{m_2})$ and
where $U=U_1\times U_2$ and $\psi=\psi_1\times \psi_2$.
The differential game form captures a differential view of the
strategic interaction between the players. Note that each player's cost
function depends on its own choice variable as well as all the other
players’ choice variables. However, each player can only affect their
payoff by adjusting their own strategy.

Critical points for the game can be characterized by the differential game form.
\begin{definition}
   A point $x\in X$ is said to be a {critical point} for the game if
    $\omega(x)=0$. 
    \label{def:criticalpt} 
\end{definition}
   In the single agent case (i.e., optimization of a single cost function),  critical points can be further classified as
    local minima, local maxima, or saddles by
    looking at second-order conditions. 
    Analogous concepts exist for continuous games.
\begin{proposition}[Proposition~2~\cite{ratliff:2013aa}]
    If $x\in X$ is a local Nash equilibrium for $(f_1,f_2)=(f,-f)$, then $\omega(x)=0$ and
    $D_{i}^2f_i(x)\geq 0$ for all $i\in \mc{I}$.
    \label{prop:necessary}
\end{proposition}
  These are necessary conditions for a local Nash
  equilibrium. 
There are also sufficient conditions for Nash equilibria. 
  Such sufficient conditions define differential Nash
  equilibria~\cite{ratliff:2016aa, ratliff:2014aa}.
  \begin{definition}
  \label{def:DNE}
  A strategy $x\in X$ is a {differential Nash equilibrium}
  for $(f_1,f_2)=(f,-f)$ if $\omega(\pt)=0$  and  $\D^2_{i}\f_i(x)$ is
  positive--definite for each $i\in\mc{I}$. 
\end{definition}
Differential Nash need not be isolated. However, if
 $d\omega(x)$ is
non-degenerate for $x$ a differential Nash, where
$d\omega=d(\psi_{X_1}\circ d\f)-d(\psi_{X_2}\circ d\f),$
then $x$ is an \emph{isolated strict local Nash equilibrium}. Intrinsically,
$d\omega\in T_2^0(X)$ is a tensor field; at a point $x$ where $\omega(x)=0$, it
determines a bilinear form constructed from the uniquely determined continuous, symmetric, bilinear forms
$\{d^2f_i(x)\}_{i=1}^n$. 
We use the notation
$D\omega$ to denote the bilinear map induced by $d\omega$ which is composed of
the partial derivatives of components of $\omega$. For example, consider a
two-player, zero-sum 
game $(\f_1,\f_2)=(f,-f)$. Then, via a slight abuse of notation, the matrix
representation of this bilinear map is given by
\[D\omega(x)=\bmat{D_{1}^2\f(x)  & D_{12}\f(x)\\  -D_{12}^T\f(x)  &
-D_{2}^2\f(x)}.\]


The following definitions are pertinent to our study of genericity and
structural stability of differential Nash equilibria; there are analogous
concepts in dynamical systems~\cite{broer:2010aa}.
    \begin{definition}
        A critical point $x$ is
    {non-degenerate} if $\det(D\omega(x))\neq 0$ (i.e.~$x$ is isolated).
    \label{def:nondegen}
    \end{definition}
    Non-degenerate differential Nash are strictly
    isolated local Nash equilibria~\cite[Theorem~2]{ratliff:2016aa}.
    \begin{definition}
        A critical point $x$ is
    {hyperbolic} if $D\omega(x)$ has no eigenvalues with zero real part.
    \label{def:hyperbolic}
    \end{definition}
All hyperbolic critical points are non-degenerate, but not all non-degenerate
critical points are hyperbolic. Hyperbolic critical points are of particular importance from the point of view of convergence, were they have local guarantees of exponential stability or instability \cite{sastry:1999ab}. We note that even in the more general manifold
setting, these definitions are invariant with respect to the coordinate
chart~\cite{ratliff:2016aa, broer:2010aa}.

%

\subsection{Mathematical Preliminaries}

In order to prove genericity of non-degenerate differential Nash, we now introduce the necessary mathematical preliminaries.

Consider smooth manifolds $X$ and $Y$ of dimension $n_x$ and $n_y$ respectively.
An \emph{$k$--jet} from $X$ to $Y$ is an equivalence class $[x,\f,U]_k$ of
triples $(x,\f,U)$ where $U\subset X$ is an open set, $x\in U$, and $f:U\rar Y$
is a $C^k$ map. The equivalence relation satisfies $[x,\f,U]_k=[y,g,V]_k$ if
$x=y$ and in some pair of charts adapted to $\f$ at $x$, $\f$
and $g$ have the same derivatives up to order $k$. We use the notation
$[x,\f,U]_k=j^k\f(x)$ to denote the $k$--jet of $f$ at $x$. The set of all
$k$--jets from $X$ to $Y$ is denoted by $J^k(X,Y)$. The jet bundle $J^k(X, Y)$
is a smooth manifold (see \cite{Hirsch:1976la} Chapter 2 for the construction).
For each $C^k$ map $f:X\rar Y$ we define a map $j^kf:X\rar J^k(X, Y)$ by $x\mapsto j^kf(x)$ and refer to it as the \emph{$k$--jet extension}. 

\begin{definition}
Let $X$, $Y$ be smooth manifolds and $f:X\rar Y$ be a smooth mapping. Let $Z$ be
a smooth submanifold of $Y$ and $p$ a point in $X$. Then \emph{$f$ intersects
$Z$ transversally at $p$} (denoted $f\pitchfork Z$ at $p$) if either $f(p)\notin
Z$ or $f(p)\in Z$ and $T_{f(p)} Y=T_{f(p)}Z + (f_\ast)_p(T_pX)$.
\end{definition}

For $1\leq k<s\leq \infty$ consider the jet map 
  $j^k:C^s(X, Y)\rar C^{s-k}(X, J^k(X, Y))$
and let $Z\subset J^k(X, Y)$ be a submanifold. 
Define
\begin{equation}
  \trans^s(X, Y; j^k, Z)=\{ h\in C^s(X, Y)|\ j^kh \pitchfork Z\}.
  \label{eq:jettrans}
\end{equation}
A subset of a topological space $X$ is \emph{residual} if it contains the intersection of countably many open--dense sets. We say a property is \emph{generic} if the set of all points of $X$ which possess this property is residual~\cite{broer:2010aa}. 



\begin{theorem}(Jet Transversality Theorem, Chap.~2 \cite{Hirsch:1976la}).
  Let $X$, $Y$ be $C^\infty$ manifolds without boundary, and let $Z\subset
  J^k(X, Y)$ be a $C^\infty$ submanifold. Suppose that $1\leq k<s\leq \infty$.
  Then, $\trans^s(X, Y; j^k, Z)$ is residual and thus dense in $C^s(X, Y)$ endowed with the strong topology, and open if $Z$ is closed. 
\label{thm:jettrans}
\end{theorem}
\begin{proposition}
  \label{prop:intersect}
  (Chap.~II.4, Proposition 4.2~\cite{golubitsky:1973aa}). Let $X,Y$ be smooth
  manifolds and $Z\subset Y$ a submanifold. Suppose that $\dim X< \codim Z$. Let
  $f:X\rar Y$ be smooth and suppose that $f\pitchfork Z$. Then, $f(X)\cap Z=\emptyset$.
\end{proposition}


The Jet Transversality Theorem and Proposition~\ref{prop:intersect} can be used to show a subset of a jet bundle having a particular set of desired properties is generic. 
Indeed, consider the jet bundle $J^k(X, Y)$ and recall that it is a manifold
that contains jets $j^kf:X\rar J^k(X, Y)$ as its elements where $f\in C^k(X,
Y)$. Let $Z\subset J^k(X, Y)$ be the submanifold of the jet bundle that \emph{does not} possess the desired properties. 
If $\dim X<\codim\ Z$, then for a generic function $f\in C^k(X, Y)$ the image of the $k$--jet extension is disjoint from $Z$ implying that there is an open--dense set of functions having the desired properties. 
It is exactly this approach we use to show the genericity of non-degenerate differential Nash equilibria of zero-sum continuous games.

\section{Theoretical Results}
\label{sec:thm}
In this section, we specialize the results in~\cite{ratliff:2013aa} and \cite{ratliff:2016aa} on genericity and structural stability of differential Nash equilibria to the class of zero-sum games. 
\subsection{Genericity}
\label{subsec:genericity}
To develop the proof that local Nash equilibria of zero-sum games are generically non-degenerate differential Nash equilibria, we leverage the fact that it is a generic property of sufficiently smooth functions that all critical points are non-degenerate. 
\begin{lemma}[{\cite[Chapter~1]{broer:2010aa}}]
    For $C^r$ functions, $r\geq 2$ on $\mb{R}^n$, or on a manifold, it is a
    generic property that all the critical points are non-degenerate.
    \label{lem:generic}
\end{lemma}

The above lemma implies that
for a generic function $f\in C^r(X,\mb{R})$ on an $m$--dimensional manifold $X$, the Hessian
\[H(x)=\bmat{D_{1}^2f(x) & \cdots & D_{1m}f(x)\\
    \vdots &\ddots & \vdots\\
D_{m1}f(x)& \cdots & D_{m}^2f(x)}\]
is non-degenerate at critical points---that is, $\det(H(x))\neq 0$. 
\begin{lemma}
    Consider $f\in C^r(X, \mb{R})$ and the zero-sum game $(f,-f)$. For any critical point $x\in X$ (i.e., $x\in\{x\in X|\ \omega(x)=0\}$), 
    $\det(H(x))\neq 0\Longleftrightarrow\det(D\omega(x))\neq 0.$
    \label{lem:equiv}
\end{lemma}

\begin{proof}
 Before proceeding  we note that in the case that $X$ is a smooth manifold, the stationarity of critical points and definiteness of $H$ and $D\omega$ are  coordinate-invariant properties and hence, independent of coordinate chart~\cite{ratliff:2013aa, ratliff:2014aa, ratliff:2016aa, broer:2010aa}. Thus, to shorten the presentation of proofs, we simply treat the Euclidean case here; showing the more general case simply requires selecting a coordinate chart defined on a neighborhood of the differential Nash, showing the results with respect to this chart, and then invoking coordinate invariance.
 
 Let $x=(x_1,x_2)$ where $X=X_1\times X_2$ and $X_i$ is $m_i$--dimensional. Note that $H(x)$ is equal to $D\omega(x)$ with the last $m_2$ rows scaled
    each by $-1$. Indeed,
    \[D\omega(x)=\bmat{D_1^2f(x) & D_{12}f(x)\\ 
    -D_{12}^Tf(x) & -D_2^2f(x)}\]
    where $D_i^2f(x)$ is $m_i\times m_i$ dimensional for each $i\in\{1,2\}$ and $D_{12}f(x)$ is $m_1\times m_2$ dimensional.
Clearly, $D\omega(x)=PH(x)$ where $P=\text{blockdiag}(I_{m_1},-I_{m_2})$ with each $I_{m_i}$ the $m_i\times m_i$ identity matrix, so that
    $\det(H(x))=(-1)^{m_2}\det(D\omega(x))$.
    Hence, the result holds. 
\end{proof}

This equivalence between the non-degeneracy of the Hessian and the game Jacobian $D\omega$ allows us to lift the fact that non-degeneracy of critical points is a generic property to zero-sum games.

\begin{proposition}
Consider a two-player, zero-sum continuous game $(f,-f)$ defined for $f\in C^r(X,\mb{R})$ with $r\geq 2$. A differential Nash equilibrium is non-degenerate, and furthermore, it is hyperbolic. 
\label{prop:hyperbolic}
\end{proposition}
\begin{proof}
     It is enough to show that all differential Nash are hyperbolic since all hyperbolic equilibria correspond to a non-degenerate $D\omega$. Further, just as we noted in the proof of Lemma~\ref{lem:equiv}, stationarity, definiteness, and non-degeneracy are coordinate-invariant properties. Thus, we simply treat the Euclidean case here.
     
     By definition, at a differential Nash equilibrium $x$ of a zero-sum game, $\omega(x)=0$, $D^2_1f(x)> 0$, and $-D^2_2f(x)>0$. Further, in zero-sum games, $D^2_{12}f=(D^2_{21}f)^T$. Thus, the bilinear map $D\omega$, takes the form
\begin{align*}
    D\omega(x)&=\bmat{D^2_1f(x) & D_{12}f(x) \\ -D_{21}f(x) & -D^2_{2}f(x)}\\
    &= \bmat{D^2_1f(x) & D_{12}f(x) \\ -D_{12}^Tf(x)& -D^2_{2}f(x)}.
\end{align*}
    Let $(\lambda,v)$ be an eigenpair of $D\omega(x)$. The real part of $\lambda$, denoted $\mathrm{Re}(\lambda)$, can be written as
\begin{align*}
\textstyle\mathrm{Re}(\lambda)&=\textstyle\frac{1}{2}(\lambda +\bar \lambda)=\textstyle\frac{1}{2} (v^*D\omega^T(x)v+v^*D\omega(x)v)\\
&=\textstyle\frac{1}{2} v^*(D\omega^T(x)+D\omega(x))v\\
&=\frac{1}{2} v^*\bmat{D^2_1f(x) &0 \\ 0& -D^2_{2}f(x)}v>0
\end{align*}
where the last line follows from the positive definiteness of $\diag(D^2_1f(x),-D_2^2f(x))$
at a differential Nash equilibrium. Hence, $x$ is hyperbolic and, clearly, at this point $\det(D\omega(x))\neq 0$.
\end{proof}

The above proposition provides a strong result for the class of zero-sum games. In particular, simply due to the structure of $D\omega$, all differential Nash have the nice property of being hyperbolic, and hence, exponentially attracting under gradient-play dynamics---that is, $\dot{x}=-\omega(x)$ or its discrete time variant $x^+=x-\gamma \omega(x)$ for appropriately chosen stepsize $\gamma$. Note that numerous learning algorithms in machine learning applications of zero-sum games take this form~(see, e.g., \cite{goodfellow:2014aa, mazumdar:2018aa, mazumdar:2019aa}). 

\begin{theorem}
    For two-player, zero-sum continuous games, non-degenerate differential Nash are generic amongst
    local Nash equilibria. That is, given a generic $f\in C^r(X, \mb{R})$, all local
    Nash equilibria of the game $(f,-f)$ are (non-degenerate) differential Nash
    equilibria. 
    \label{thm:main}
\end{theorem}
\begin{proof}
    First, critical points of a function $f$ are those such that
    $(D_1f_1(x)\ D_2f_2(x))=0$ and hence they coincide with critical points of
    the zero-sum game---i.e., those points $x$ such that $\omega(x)=(D_1f(x),
    -D_2f(x))=0$. By Lemma~\ref{lem:equiv}, for any critical point $x$, $\det(H(x))=0$ if and
    only if $\det(D\omega(x))=0$. Hence, critical points of $f$ are
    non-degenerate if and only if critical points of the zero-sum game are
    non-degenerate. 

Consider a generic function $f$ and the
    corresponding zero-sum game $(f,-f)$. If $X$ is a smooth manifold, let $(U, \varphi)$ be a product chart on $X_1\times X_2$ that contains $x$. Suppose that $x$ is a local Nash
    equilibrium so that $\omega(x)=0$ and $D_1^2f(x)\geq 0$ and $-D_2^2f(x)\geq
    0$. By the above argument, since $f$ is generic and the critical points of
    $f$ coincide with those of the zero-sum game, $\det(D\omega(x))\neq 0$.   By Lemma~\ref{lem:generic}, critical points of a generic zero-sum
    game are non-degenerate. That is, there exists an open-dense set of
    functions $f$ in $C^r(X, \mb{R})$ such that critical points of the
    corresponding game are non-degenerate.

  Let $J^2(X, \mb{R})$ denote the second-order jet bundle containing $2$--jets $j^2f$ such that $f:X\rar \mb{R}$. Then, $J^2(X, \mb{R})$ is locally diffeomorphic to 
  \[\mb{R}^m\times \mb{R}\times \mb{R}^m\times \mb{R}^{\frac{m(m+1)}{2}}\]
  and the $2$--jet extension of $f$ at any point $x\in X$ in coordinates is given by
  \[(\varphi(x), (f\circ\varphi^{-1})(\varphi(x)), D^\varphi f(x), (D^\varphi)^2f(x))\]
  where $D^\varphi f=[D^\varphi_1f\ D^\varphi_2f]$ with $D^\varphi_j=[\partial (f\circ \varphi^{-1})/(\partial y_j^1)\ \cdots \ \partial (f\circ \varphi^{-1})/(\partial y_j^{m_i})]$ and similarly for $(D^\varphi)^2f$. Again, we note that the properties of interest (stationarity, definiteness, and non-degeneracy) are known to be coordinate invariant. 
  
  Consider a subset of $J^2(X, \mb{R})$ defined by 
\[\mc{D}=\mb{R}^m\times \mb{R}\times \{0_m\}\times Z(m_1)\times \mb{R}^{m_1\times
m_2}\times Z(m_2)\]
where $Z(m_i)$ is the subset of symmetric $m_i\times m_i$ matrices such that for
$A\in Z(m_i)$, $\det(A)=0$. Each $Z(m_i)$ is algebraic and has no interior
points; hence, we can use the Whitney stratification theorem~\cite[Chapter 1, Theorem 2.7]{gibson:1976ab} to get that each
$Z(m_i)$ is the union of submanifolds of co-dimension at least $1$. Hence $\mc{D}$
is the union of submanifolds and has co-dimension at least $m+2$. Applying the
Jet Transversality Theorem (Theorem~\ref{thm:jettrans}) and
Proposition~\ref{prop:intersect} yields an open-dense set of functions $f$ such
that when $\omega(x)=0$, $\det(D_i^2f(x))\neq 0$, for $i=1,2$.

Now, the intersection of two open-dense sets is open-dense so that we have an
open-dense set of functions $f$ in $C^r(X, \mb{R})$ such that when
$\omega(x)=0$, $\det(D_i^2f(x))\neq 0$ for each $i\in\{1,2\}$ and
$\det(D\omega(x))\neq 0$. 
This, in turn, implies that there is an open-dense set $\mc{F}$ of functions $f$ in
$C^r(X,\mb{R})$ such that for zero-sum games constructed from these
functions, local Nash equilibria are non-degenerate differential Nash
equilibria. Indeed, consider an $f\in \mc{F}$ in this set such that $x$ is a
local Nash equilibrium of $(f,-f)$. Then necessary conditions for Nash imply
that $\omega(x)=0$, $D_1^2f(x)\geq 0$ and $-D_2^2f(x)\geq 0$. However, since
$f\in \mc{F}$, $\det(D_1^2f(x))\neq 0$ and $\det(-D_2^2f(x))=(-1)^{m_2}\det(D_2^2f(x))\neq
0$. Hence, $x$ is a differential Nash equilibrium. Moreover, since $f\in
\mc{F}$, $\det(H(x))\neq 0$ which is equivalent to $\det(D\omega(x))\neq 0$ (by
Lemma~\ref{lem:equiv}). Thus, $x$ is a non-degenerate differential Nash. 
\end{proof}

As shown in Proposition~\ref{prop:hyperbolic}, all differential Nash for zero-sum games are non-degenerate simply by the structure of $D\omega$. This further implies that local Nash equilibria are generically hyperbolic critical points, meaning there are no eigenvalues of $D\omega$ with zero real part.
\begin{corollary}
    Within the class of two-player zero-sum continuous games, local Nash equilibria are generically hyperbolic critical points.
    \label{cor:hyperbolic}
\end{corollary}
\begin{proof}
  Consider a two-player, zero-sum game $(f,-f)$ for some generic sufficiently smooth $f\in C^r(X, \mb{R})$. Then, by Theorem~\ref{thm:main}, a local Nash equilibria $x$ is a differential Nash equilibria. Moreover,  by Proposition~\ref{prop:hyperbolic}, $x$ is hyperbolic so that all eigenvalues of $D\omega(x)$ must have strictly positive real parts. This implies that all such points are hyperbolic critical points of the gradient dynamics $\dot{x}=-\omega(x)$. 
\end{proof}

\subsection{Structural Stability}
\label{subsec:structuralstability}
Genericity gives a formal mathematical sense of 'almost all' for a certain property---in this case, non-degeneracy and further hyperbolic. In addition, we show that (non-degenerate) differential Nash are structurally stable, meaning that they persist under smooth perturbations within the class of zero-sum games.
\begin{theorem}
    For zero-sum games, differential Nash equilibria are {structurally stable}: given $f \in C^r(X_1 \times X_2,\mb{R})$, $g\in C^r(X_1 \times X_2,\mb{R})$, and a differential Nash equilibrium $(x_1,x_2) \in X_1 \times X_2$, there exists a neighborhoods $U \subset \mb{R}$ of zero and $V \subset X_1 \times X_2$ such that for all $t\in U$ there exists a unique differential Nash equilibrium $(\tilde x_1,\tilde x_2)  \in V$ for the zero-sum game $(f+tg,-f-tg)$. 
    \label{thm:ss_zsg}
\end{theorem}

\begin{proof}
	Define the smoothly perturbed cost function $\tilde f: X_1\times X_2 \times \mb{R} \rar \mb{R}$ by
	$\tilde f(x,y,t)=f(x,y)+tg(x,y)$,
 	and its differential game form $\tilde \omega:  X_1\times X_2 \times \mb{R} \rar T^*(X_1 \times X_2)$ by
 	\[ \tilde \omega(x,y,t)=(D_1(\tilde f(x,y)+tg(x,y),-D_2(\tilde f(x,y)+tg(x,y)),\]
 	for all $t \in \mb{R}$ and $(x,y) \in X_1 \times X_2$. 

 	Since $(x_1,x_2)$ is a differential Nash equilibrium, $D\tilde\omega(x,y,0)$ is necessarily non-degenerate (see the proof of Corollary~\ref{cor:hyperbolic}). Invoking the implicit function theorem~\cite{lee:2012aa}, there exists neighborhoods $V\subset \mb{R}$ of zero and $W \subset X_1 \times X_2$ and a smooth function $\sigma \in C^r(V,W)$ such that for all $ t \in V$ and $(x_1,x_2) \in W$, \[\tilde \omega(x_1,x_2,s)=0 \iff (x_1,x_2)=\sigma(t).\]

 	Since $\tilde \omega$ is continuously differentiable,
 	there exists a neighborhood $U \subset W$ of zero such that $D\tilde \omega(\sigma(t),t)$ is invertible for all $t \in U$. Thus, for all $t \in U$, $\sigma(t)$ must be the unique Nash equilibrium of  $(f+tg|_W,-f-tg|_W)$. 
\end{proof}
We note that both the genericity and structural stability results follow largely from the fact that the class of two-player zero-sum games are defined completely in terms of a single (sufficiently) smooth function $f\in C^r(X,\mb{R})$, so that its fairly straightforward to lift the properties of genericity and structural stability to the class of zero-sum games from the class of smooth functions. We also remark that the perturbations considered here are those such that the game remains in the class of zero-sum games; that is, the function $f$ is smoothly perturbed and this induces the perturbed zero sum game $(f+tg,-f-tg)$.

\section{Examples}
\label{sec:num_ex}



To illustrate the implications of structural stability, we provide a simple example. Consider a classic set of zero-sum continuous games known as biliear games. Such games have similar characteristics as bimatrix games played on the simplex; in particular, bimatrix games have the same cost structure as bilinear games where the stratgegy space of the former is considered to be a probability distribution over the finite set of pure strategies. This is particularly interesting since it demonstrates that interior equilibria of such games can be altered arbitrarily small perturbations. 

\begin{example}
    Consider two-players with decision variables $x \in\mb{R}^{d_x}$ and $y \in \mb{R}^{d_y}$ respectively, playing a zero-sum game on the function:
    \[f(x,y)=x^T A y\]
    Where $A \in \mb{R}^{d_x \times d_y}$. The $x$ player would like to minimize $f$ while the $y$ player would like to maximize it. Looking at $\omega$ for this game, we can see that the local Nash equilibria live in $\mc{N}(A) \times \mc{N}(A^T)$, where $\mc{N}(A)$ and $\mc{N}(A^T)$ denote the nullspaces of $A$ and $A^T$ respectively:
    \begin{align*} 
    \omega(x,y) &= \bmat{Ay \\ -A^Tx}
    \end{align*}
     We note that the local Nash equilibria are not differential Nash equilibria, and that $D\omega$ has purely imaginary eigenvalues everywhere since it is skew-symmetric. Thus the local Nash equilibria are non-hyperbolic and this a non-generic case. Letting $f_\epsilon=f(x,y)-\frac{\epsilon}{2} ||x||^2$, we see that $\omega$ for this perturbed game (denoted $\omega_\epsilon$) has the form:
    \begin{align*} 
    \omega_\epsilon(x,y) &= \bmat{Ay - \epsilon x \\ -A^Tx}
    \end{align*}
    This perturbation fundamentally changes the critical points, and looking at $D\omega_\epsilon$, we can see that for any $\epsilon>0$, there are no more local Nash equilibria:
    \begin{align*} 
    D\omega_\epsilon(0,0) &= \bmat{ - \epsilon I_{d_x} & A \\ -A^T & 0}
    \end{align*}
    Since any arbitrarily small perturbation of this form can cause all of the local Nash equilibria to change, these games cannot be structurally stable.
\end{example}

We now show how this behavior extends to more complicated settings. Specifically we present an example of a game of rock-paper-scissors where both players have stochastic policies over the three actions which are parametrized by weights. The following example highlights how this classic problem is non-generic and the behavior changes drastically when the loss is perturbed in a small way.

\begin{example}
    Consider the game of rock-paper-scissors where each player has three actions $\{0,1,2\}$, with payoff matrix: 
    \[ M=\bmat{\ 0 & -1 & \ 1 \\ \ 1 & \ 0 & -1 \\ -1 & \ 1 & \ 0}\]
    Each player $i \in \{1,2\}$ has a policy or mixed strategy  $\pi_i$ parametrized by a set of weights $\{w_{ij}\}_{j \in \{0,1,2\}}$ of the form:
    \[ \pi_i(j)=\frac{\exp(-\beta_i w_{ij})}{\sum_{k=0}^2 \exp(-\beta_i w_{ij})}\]
    Where $\beta_i$ is a hyper-parameter for player $i$ that determines the 'greediness' of their policy with respect to their set of weights. For simplicity, we treat $\pi_i$ as a vector in $\mb{R}^3$. Each player would like to maximize their expected reward given by
    \[ f(w_1,w_2)= \pi_1^T M \pi_2. \]
    We note that there is a continuum of local Nash equilibria for the policies $\pi_i=[\frac{1}{3},\frac{1}{3},\frac{1}{3}]$ for $i \in \{1,2\}$ and that this is achieved whenever each player has all of their weights equal. 
\end{example}

\begin{figure}[h]
\begin{center}
  \includegraphics[width=0.75 \columnwidth]{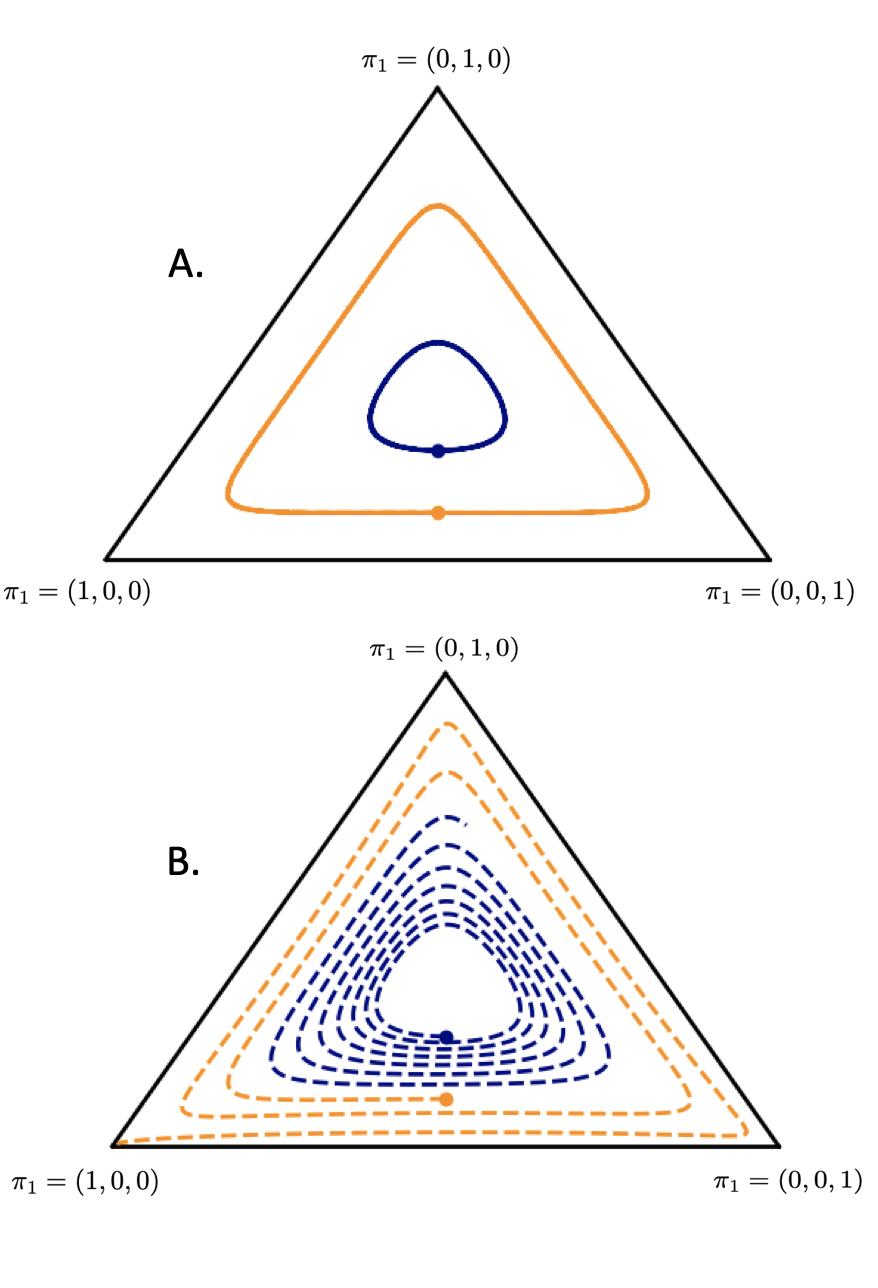}
  \caption{ The trajectory of the policy of player 1 under gradient-play for A. rock-paper scissors and B. a perturbed version of rock-paper-scissors.  A. Player 1 cycles around the local Nash equilibrium of $(\frac{1}{3},\frac{1}{3},\frac{1}{3})$ from either initialization (shown with circles). We remark that player 1's time average policy is in fact $(\frac{1}{3},\frac{1}{3},\frac{1}{3})$. B. Player 1 diverges from the local Nash equilibrium from either initialization for the perturbed game given by \eqref{eq:rps_perturbed}.}
  \label{fig:rps}
  \end{center}
\end{figure}

In Fig.~\ref{fig:rps} we show the trajectories of the policy of player $1$, when $\beta_1=\beta_2=1$ and both players use gradient descent to update their weights at each iteration. In Figure~\ref{fig:rps}A. we see that player 1 cycles around the local Nash equilibrium in policy space. In Figure~\ref{fig:rps}B. we show the trajectories of the policy of player $1$, starting from the same initializations, but for a perturbed version of the game defined by
\begin{align}
    f_{\epsilon}(w_1,w_2)= \pi_1^T M \pi_2 +\epsilon g(w_1,w_2)
    \label{eq:rps_perturbed}
\end{align}
 where $\epsilon= 1$e-$3$ and $g(x,y)=||y||^2- ||x||^2$. Here we can see that this relatively small perturbation causes a drastic change in the behavior where player 1 diverges from the Nash of the original game and converges to the sub-optimal policy of always playing action zero.

\section{Discussion and Concluding Remarks}
\label{sec:conclusion}
The focus of this paper is on the genericity and structural stability of a particular refinement of the local Nash equilibrium concept---namely, differential Nash equilibria---within the class of two-player, zero-sum continuous games. 
The renewed interest in zero-sum games on continuous action spaces is primarily due to the widespread adoption of game theoretic tools in areas such as robust reinforcement learning and adversarial learning including generative adversarial networks. 
For instance, zero-sum continuous game abstractions have shown to be particularly adept at learning robust policies for a wide-variety of tasks from classification to prediction to control.

Most learning approaches are based on local information such as gradient updates, and as such, representations of Nash equilibria that are amenable to computation such as the differential Nash concept are extremely relevant. Much of the existing convergence analysis for machine learning algorithms based on game-theoretic concepts proceeds under the structural assumptions implicit in the definition of the differential Nash equilibrium concept. In this paper, we show that characterizations such as these are generic and structurally stable; hence, the aforementioned structural assumptions only rule out a measure zero set of games, and the desired properties are robust to smooth perturbations in player costs.

\appendix
\subsection{Additional Mathematical Preliminaries}
\label{app:prelims}
In this appendix, we provide some additional mathematical preliminaries; the interested reader should see standard references for a more detailed introduction~\cite{lee:2012aa, abraham:1988aa}.

A \emph{smooth manifold} is a topological
manifold with a smooth atlas. In particular, we use the term \emph{manifold} generally; we specify whether it
is a finite-- or infinite--dimensional manifold only when necessary.
If a covering by charts takes their values in a
Banach space $E$, then $E$ is called the \emph{model space} and we say that $X$
is a $C^r$--\emph{Banach manifold}. 
For a vector space $E$, we define the vector space of continuous
$(r+s)$--multilinear maps $T_s^r(E)=L^{r+s}(E^\ast, \ldots, E^\ast, E,
\ldots, E; \mb{R})$ with $s$ copies of $E$ and $r$ copes of $E^\ast$ and where $E^\ast$ denotes the dual. 
Elements of $T_s^r(E)$ are \emph{tensors} on $E$, 
and
$T^r_s(X)$ 
denotes the 
vector bundle of such tensors%
~\cite[Definition~5.2.9]{abraham:1988aa}. 

 Suppose $f:X\rar M$ is a mapping of one
 manifold $X$ into another $M$.
 Then, we can interpret
 the derivative of $f$ on each chart at $\pt$ as a linear mapping 
   $df(\pt):T_\pt X\rar T_{f(\pt)}M.$
 When $M=\mb{R}$, the collection of such
maps defines a \emph{$1$--form} $df:X\rar T^\ast X$. 
 Indeed, a
$1$--form is a continuous map $\omega:X\rar T^\ast X$ satisfying $\pi\circ
\omega=\text{Id}_X$ where $\pi:T^\ast X\rar X$ is the natural projection mapping
$\omega(x)\in T^\ast_x X$ to $x\in X$.

At a critical point $\pt\in X$ (i.e., where $df(x)=0$), 
there is a uniquely determined continuous, symmetric bilinear form 
$d^2f(\pt)\in T_2^0(X)$ such that $d^2f(\pt)$
is defined for all $v,w\in T_xX$ by
  $d^2(f\circ \vphi^{-1})(\vphi(\pt))(v_\vphi, w_\vphi)$
where $\vphi$ is any product chart at $x$ and $v_\vphi, w_\vphi$ are the local representations
of $v,w$ respectively~\cite[Proposition in \S 7]{palais:1963aa}. We say $d^2f(\pt)$ is
\emph{positive semi--definite} if there exists $\alpha\geq 0$ such that for any chart
$\vphi$,
\begin{equation}
d^2(f\circ \vphi^{-1})(\vphi(x))(v,v)\geq \alpha \|v\|^2, \ \ \forall \ v\in
  T_{\vphi(x)}E.
  \label{eq:posdef}
\end{equation}
If $\alpha>0$, then we say $d^2f(\pt)$ is \emph{positive--definite}.
Both critical points and positive definiteness are invariant
with respect to the choice of coordinate chart.

Consider smooth manifolds $X_1,X_2$. 
The product space $X_1\times X_2$ is naturally a smooth
manifold~\cite[Definition~3.2.4]{abraham:1988aa}. 
 There is a canonical isomorphism at each point such that the cotangent bundle
 of the product manifold splits:
  \begin{equation} 
    T^\ast_{(\pt_1,x_2)}(X_1\times X_2)\cong
    T^\ast_{\pt_1} X_1\oplus
    T^\ast_{\pt_2} X_2
    \label{eq:canon}
  \end{equation}
  where $\oplus$ denotes the direct sum of vector spaces.
 There are natural bundle maps $\psi_{X_1}:T^\ast(X_1\times X_2)\rar T^\ast(X_1\times X_2)$ annihilating the all the components other
 than those corresponding to $X_i$ of an element in
 the 
 cotangent bundle. 
   In particular, 
   $\psi_{X_1}(\omega_1, \omega_2)=(0_1, \omega_2)$ and $\psi_{X_2}(\omega_1,\omega_2)=(\omega_1,0_2)$
 where $\omega=(\omega_1,\omega_2)\in T^\ast_x(X_1\times
 X_2)$ and $0_j\in T^\ast_{x_j} X_j$ for each $j\in\{1,2\}$ is
 the zero functional.
 

\bibliographystyle{IEEEtran}
\bibliography{main}
\end{document}